\documentclass[journal]{IEEEtran}
\usepackage[table]{xcolor}
\usepackage{times}
\usepackage{epsfig}
\usepackage{amsmath}
\usepackage{amsthm}
\usepackage{amsfonts}
\usepackage{graphicx}
\usepackage{amssymb}
\usepackage{amstext}
\usepackage{latexsym}
\usepackage{color,colortbl}
\usepackage{ifthen}
\usepackage{multirow}
\usepackage{verbatim}
\usepackage{array,tabularx}
\usepackage{arydshln}
\usepackage[mathscr]{euscript}
\usepackage{accents}
 \usepackage{cite}
\usepackage{hhline}
\usepackage{caption}
\usepackage{subcaption}
\usepackage{enumerate}
\usepackage{xcolor}
\usepackage{mathtools}
\usepackage{url}
\usepackage{xparse}
\usepackage{makecell}
\usepackage{varwidth}
\usepackage{bm}
\usepackage{arydshln}
\usepackage{comment}
\usepackage{wrapfig}

\usepackage{caption}
\usepackage{float}
\usepackage{booktabs}

\usepackage{nidanfloat}
\usepackage{subcaption}

\usepackage{multirow}

\newcolumntype{C}[1]{>{\centering\let\newline\\\arraybackslash\hspace{0pt}}m{#1}}

\usepackage[normalem]{ulem}

\usepackage{amsmath,pgfplots,amssymb}
\usepackage{subcaption,graphicx}

\newtheorem{theorem}{Theorem}
\newtheorem{restate}{Theorem}
\newtheorem{lemma}{Lemma}

\newtheorem{proposition}{Proposition}

\theoremstyle{definition}
\newtheorem{example}{Example}
\newtheorem{definition}{Definition}

\theoremstyle{definition}
\newtheorem{remark}{Remark}

\theoremstyle{definition}
\newtheorem*{construction}{Construction}

\DeclareMathOperator{\lcm}{lcm}
\DeclareMathOperator{\lift}{lift}

\usetikzlibrary{matrix,decorations.pathreplacing}

\ifCLASSINFOpdf
\else
\fi
\begin{document}
%
\title{One-Shot PIR: Refinement and Lifting}
%
%
%

\author{Rafael G.L. D'Oliveira, Salim El Rouayheb \thanks{This work was supported in part by NSF Grant CCF 1817635 and presented in part at the IEEE International Symposium on Information Theory 2018.} \\ ECE, Rutgers University, Piscataway, NJ\\ Emails: rafael.doliveira@rutgers.edu, salim.elrouayheb@rutgers.edu }

\maketitle

\begin{abstract}
We study a class of private information retrieval (PIR) methods that we call one-shot schemes. The intuition behind one-shot schemes is the following. The user's query  is regarded as a   dot product of a query vector and the message  vector (database) stored at multiple servers.
 
Privacy, in an information theoretic sense, is then achieved by encrypting the query vector using a secure linear code, such as secret sharing. 

Several PIR schemes in the literature, in addition to novel ones constructed here, fall into this class.
One-shot schemes provide an insightful link between  PIR and data security  against eavesdropping. However, their download rate is not optimal, i.e., they do  not achieve the PIR capacity. 
Our main contribution is two transformations of one-shot schemes, which we call refining and lifting.  We show that refining and lifting one-shot schemes gives capacity-achieving schemes for the cases when the PIR capacity is known. In the other cases, when the PIR capacity is still unknown, refining and lifting one-shot schemes gives the best download rate so far.

\end{abstract}

%
\IEEEpeerreviewmaketitle

\section{Introduction}

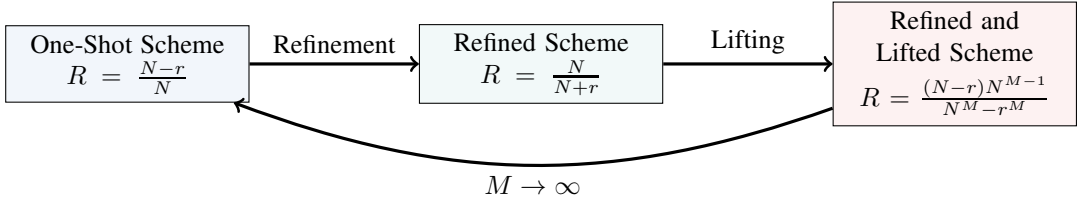
\begin{figure*}[!t]
    \centering
    
    \definecolor{bostonuniversityred}{rgb}{0.8, 0.0, 0.0}
    \definecolor{darkcyan}{rgb}{0.0, 0.55, 0.55}
    \definecolor{cobalt}{rgb}{0.0, 0.28, 0.67}

\tikzstyle{int1}=[draw, fill=cobalt!5, minimum size=2em,  text width=3cm,align=center]
\tikzstyle{int2}=[draw, fill=darkcyan!5, minimum size=2em,  text width=3cm,align=center]
\tikzstyle{int3}=[draw, fill=bostonuniversityred!5, minimum size=2em,  text width=3cm,align=center]
\tikzstyle{init} = [pin edge={to-,thin,black}]

\begin{tikzpicture}[node distance=5.5cm,auto]
    \node [int1] (a) {One-Shot Scheme \\ $R = \frac{N-r}{N}$};
    \node [int2] (c) [right of=a] {Refined Scheme \\ $R = \frac{N}{N+r}$};
    \node [int3] (d) [right of=c] {Refined and Lifted Scheme \\[5pt] $R=\frac{(N-r)N^{M-1}}{N^M - r^M}$};

    \path[->,line width=1.2pt] (a) edge node {Refinement} (c);
    \path[->,line width=1.2pt] (c) edge node {Lifting} (d);
    \path[->,bend left=20,line width=1.2pt] (d) edge node {$M \rightarrow \infty$} (a);
    
\end{tikzpicture}

\caption{Given a One-shot scheme we can apply the Refinement Lemma to obtain a refined scheme with a better rate for $M=2$ messages. Lifting this scheme, through the Lifting Theorem, we obtain a refined and lifted scheme with better rate for any number, $M$, of messages. When $M \rightarrow \infty$ the rate of the refined and lifted scheme converges to that of the one-shot scheme.}

\label{fig:results}

\end{figure*}

%
%
%
%
\IEEEPARstart{W}{e} consider the  problem of designing  private information retrieval (PIR) schemes on coded data stored on multiple servers that can possibly collude. In this setting, a user wants to download a message from a server with $M$ messages while revealing no information about which message it is interested in. The database is  replicated on $N$ servers, or, in general, is stored using an erasure code, typically a Maximum Distance Separable (MDS) code. Some of these servers could possibly collude to gain information about the identity of the user's retrieved message. 

The PIR problem was first introduced and studied in \cite{PIR1995} and \cite{chor1998private} and was followed up by a large body  of work (e.g. \cite{gasarch2004survey,sun2016capacitynoncol, sun2016capacity, yekhanin2010private, beimel2001information, beimel2002breaking}). The model there assumes the database to be replicated and focuses on PIR schemes with efficient total communication rate, i.e., upload and download. Motivated by big data applications and recent advances in the theory of codes for distributed storage, there has been a growing interest in designing PIR schemes that can query data that is stored in coded form and not just replicated. For this setting, the assumption has been that the messages being retrieved are very large (compared to the queries) and therefore the focus has been on designing PIR schemes that minimize the  download rate.  Despite  significant recent progress, the problem of characterizing the optimal PIR download rate (called PIR capacity) in the case of coded data and server collusion remains open in general. \mbox{}\\

\noindent{\em{Related work:}} For replicated data,  the problem of finding the PIR capacity,  i.e., minimum download rate, is solved. It was shown in \cite{sun2016capacitynoncol} and \cite{sun2016capacity} that the  PIR capacity  is 
\begin{align} \label{eq:jafarcapacity}
    C = \frac{(N-T)N^{M-1}}{N^M - T^M} ,
\end{align}
where $N$ is the number of servers, $T$ is the number of colluding servers and $M$ is the number of messages. Capacity achieving PIR schemes were also presented in \cite{sun2016capacitynoncol} and \cite{sun2016capacity}. 

When the data is coded  and stored on a large number of servers (exponential in the number of messages), it was shown  in \cite{shah2014one} that downloading one extra bit is enough to achieve privacy. In \cite{chan2014private}, the authors derived bounds on the tradeoff between storage cost and download cost for linear coded data and studied properties of PIR schemes on MDS data. Explicit constructions of efficient PIR scheme on MDS data were first presented in  \cite{tajeddine2016private} for both collusions and no collusions. Improved PIR schemes for MDS coded data with collusions were presented in \cite{freij2016private}. PIR schemes for general linear codes, not necessarily MDS, were studied in \cite{kumar2017private}. The PIR capacity for MDS coded data and no collusion was determined in  \cite{banawan2016capacity}, and remains unknown for the case of collusions, the topic of this paper. Table \ref{table:known} summarizes the PIR capacity results known in the literature prior to this work. 

\begin{table}[H]
\centering
\begingroup
\renewcommand{\arraystretch}{2}
\begin{tabular}{l l l }
 & No Collusion &  $T$-Collusion \\ \toprule
Replication &   $C = \frac{N^M - N^{M-1}}{N^M - 1}$  &  $C = \frac{(N-T)N^{M-1}}{N^M - T^M}$    \\[5pt]
$(N,K)$-MDS & $C= \frac{(N-K)N^{M-1}}{N^M - K^M}$ & $R = \frac{N-K-T+1}{N}$ \\ \bottomrule \vspace{-5pt}
\end{tabular}
\endgroup

\caption{Summary of the best known rates in the literature, where $C$ denotes that the rate is capacity achieving.}
\label{table:known}
\end{table}

 One of the main results of this work is improving the achievable PIR rate for coded date with collusion from the rate $R = \frac{N-K-T+1}{N}$ in the bottom right corner of Table \ref{table:known} to $R=\frac{(N-K-T+1)N^{M-1}}{N^M - (K+T-1)^M}$ in Theorem \ref{teo:holl}. Before elaborating more on this, we need to introduce the class of one-shot PIR schemes. \mbox{}\\

\noindent{\em One-Shot Schemes:} A natural approach to PIR is to hide the queries from the servers using a linear secure code. We revisit the two-server toy-example  in \cite{PIR1995}, but interpret it as a Shannon one-time pad scheme \cite{shannon1949communication}.

\begin{example}[One-Time Pad PIR Scheme] \label{ex: intro}

Suppose $M$ single bit messages, $\bm{D}_1, \ldots, \bm{D}_M \in \mathbb{F}_2$, are stored, replicated, onto two non-colluding servers and a user is interested in retrieving $\bm{D}_i$ privately. The user can use the following scheme based on the one-time pad, where $\langle \bm{D} , \bm{q} \rangle$ denotes the inner product between the vectors $\bm{D}, \bm{q} \in \mathbb{F}_2^M$.

\begin{table}[H]
\centering
\begin{tabular}{c c c }
 & Server $1$ &  Server $2$ \\
\toprule
Queries & $\bm{q}$ & $\bm{q}+\bm{e}_i$   \\
Responses & $\langle \bm{D} , \bm{q} \rangle$ & $\langle \bm{D} , \bm{q}+\bm{e}_i \rangle$ \\ 
\end{tabular}
\caption{Query and response structure for hiding queries.}
\label{table:intro}
\end{table}

The queries in Table \ref{table:intro} satisfy the following properties:
\begin{itemize}
\item The vector $\bm{q} \in \mathbb{F}^M_2$ is chosen uniformly at random.
\item The vector $\bm{e}_i$ is the $i$-th vector of the standard basis of $\mathbb{F}_2^M$.
\end{itemize}

The user can retrieve $\bm{D}_i$ by summing both responses,
\begin{align} \label{equ:intro decoding}
\langle \bm{D} , \bm{q} \rangle + \langle \bm{D} , \bm{q}+\bm{e}_i \rangle = \langle \bm{D} , \bm{e}_i \rangle = \bm{D}_i .  
\end{align}
The scheme is private since, from any of the servers' point of view, the query they received is a uniformly random vector. The PIR rate of this scheme is $1/2$ since the user downloaded $2$ bits, $\langle \bm{D} , \bm{q} \rangle$ and $\langle \bm{D} , \bm{q}+\bm{e}_i \rangle$, to retrieve $1$ bit, $\bm{D}_i$.
\end{example}

In Example \ref{ex: intro}, the user sends two types of queries to the servers: ``noise'' queries, which provide no useful information to the user, like the query, $\bm{q}$, sent to Server 1, and  ``mixed'' queries, which hide the information being retrieved from the server, like the query, $\bm{q}+\bm{e}_i$, sent to Server 2.

A one-shot scheme follows this same structure. The user sends ``noise'' queries to the first $r$ servers and  ``mixed'' queries to the remaining $N-r$ servers. We call the parameter, $r$, denoting the number of ``noise'' queries, the codimension of the one-shot scheme. The PIR rate of a one-shot scheme is determined by its codimension and is given by $\frac{N-r}{N}$. \mbox{}\\

\noindent{\em{Contributions:}} In what follows, $N$ is the number of servers, any $T$ of which may collude, $M$ is the number of messages, and the messages are stored as an $(N,K)$-MDS code.

Our main result is a constructive combinatorial procedure, that we call refinement and lifting, presented in Sections \ref{sec:refine} and \ref{sec:lift}, which improves the rate of any one-shot scheme.

\begin{theorem}[The Lifting Theorem] \label{teo: main}
Any one-shot scheme with co-dimension $r$ and, consequently, rate $(N-r)/N$ can be refined and lifted to a PIR scheme with rate \[ R = \frac{(N-r)N^{M-1}}{N^M - r^M} > \frac{N-r}{N} .\]
\end{theorem}

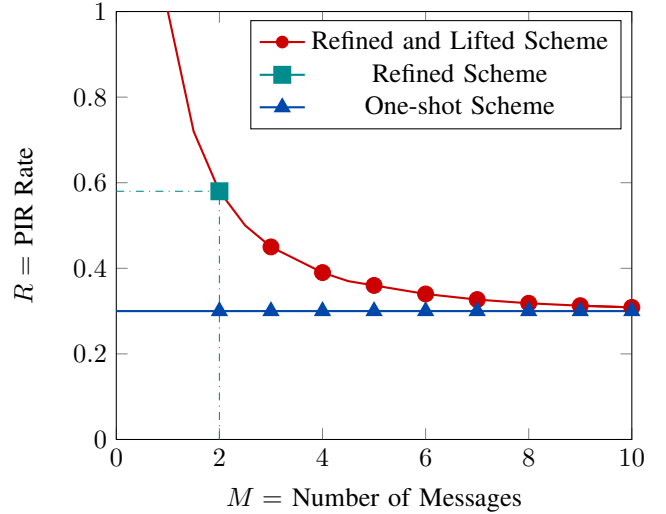
\begin{figure}[H]
    \centering
    \definecolor{bostonuniversityred}{rgb}{0.8, 0.0, 0.0}
    \definecolor{darkcyan}{rgb}{0.0, 0.55, 0.55}
    \definecolor{cobalt}{rgb}{0.0, 0.28, 0.67}
    
\begin{tikzpicture}
\begin{axis}[xlabel={$M=$ Number of Messages}, ylabel={$R=$ PIR Rate},xmin=0, xmax=10,ymin=0,ymax=1, samples=50]

  \addplot[bostonuniversityred,   thick, domain=1:20, mark = *, mark options = {solid}, mark size = 2] coordinates{(9,0.3126)(10,0.3087)};
  \addlegendentry{Refined and Lifted Scheme}

  \addplot[bostonuniversityred,   thick, domain=1:20 , forget plot] coordinates{(0.5,1.83)(1,1)(1.5,0.72)(2,0.58)(2.5,0.50)(3,0.45)(3.5,0.42)(4,0.39)(4.5,0.37)(5,0.36)(5.5,0.35)(6,0.34)(7,0.3269)(8,0.3183)(9,0.3126)(10,0.3087)};

  \addplot[bostonuniversityred, domain=1:20, mark = *, mark options = {solid}, mark size = 3, only marks, forget plot] coordinates{(3,0.45)(4,0.39)(5,0.36)(6,0.34)(7,0.3269)(8,0.3183)(9,0.3126)(10,0.3087)};
  
  \addplot[darkcyan,   thick, mark=square* , mark options = {solid}, mark size = 3] coordinates{(2,0.58)};
  \addlegendentry{Refined Scheme}
  
  \addplot[darkcyan,  dashdotted, forget plot] coordinates{(2,0)(2,0.58)};

  \addplot[darkcyan,  dashdotted, forget plot] coordinates{(0,0.58)(2,0.58)};

  \addplot[cobalt,   thick, domain=0:20, mark = triangle*, mark options = {solid}, mark size = 3 ] coordinates{(2,0.3)(3,0.3)(4,0.3)(5,0.3)(6,0.3)(7,0.3)(8,0.3)(9,0.3)(10,0.3)};
  \addlegendentry{One-shot Scheme}
  
  \addplot[cobalt,thick, domain=0:20, forget plot ] coordinates{(0,0.3)(2,0.3)};

\end{axis}
\end{tikzpicture}
\caption{The schemes in Figure \ref{fig:results} for $N=10$ and $r=7$.}
\end{figure}
More so, refining and lifting one-shot schemes, we obtain capacity-achieving schemes for the cases when the PIR capacity is known and the best known rates for when the capacity is not known.

Applying Theorem \ref{teo: main} to the one-shot scheme in \cite{freij2016private}, we obtain the first PIR scheme to achieve the rate conjectured to be optimal\footnote{The optimality of the rate in Equation \ref{eq:hollanti intro} was disproven in  \cite{jafar2018conjecture} for some parameters. For the remaining range of parameters, the PIR schemes obtained here, in Theorem \ref{teo:holl}, through refining and lifting, achieve the best rates known so far in the literature.} in \cite{freij2016private} for MDS coded data with collusions.

\begin{theorem} \label{teo:holl}
Refining and lifting the scheme presented in \cite{freij2016private} we obtain a PIR scheme with rate
\begin{align} \label{eq:hollanti intro}
R = \frac{(N-K-T+1)N^{M-1}}{N^M - (K+T-1)^M} .
\end{align}
\end{theorem}

Another interesting consequence of our procedure is that applying it to a $T$-threshold linear secret sharing scheme we obtain the same, capacity-achieving, rate as the scheme presented in \cite{sun2016capacity} but with less queries. 

\begin{theorem}
Refining and lifting a $T$-threshold linear secret sharing scheme we obtain a PIR scheme with capacity-achieving rate
\[ R_{\mathcal{Q}} = \frac{(N-T)N^{M-1}}{N^M - T^M} .\]
\end{theorem}

\section{Setting}

A set of $M$ messages, $\{ \bm{W}^1, \bm{W}^2, \ldots, \bm{W}^M \} \subseteq \mathbb{F}_q^L$, are stored on $N$ servers each using an $(N,K)$-MDS code. We denote by $ \bm{W}_i^j \in \mathbb{F}_q^{L/K}$, the data about $\bm{W}^j$ stored on server~$i$.

\begin{table}[H]
\centering
\begingroup
\renewcommand{\arraystretch}{1.5}
\begin{tabular}{ c  c  c  c}
Server $1$ & Server $2$ &$\cdots$ & Server $N$ \\
\toprule
 $\bm{W}_1^1$ & $\bm{W}_2^1$ &$\cdots$ & $\bm{W}_N^1$ \\
 $\bm{W}_1^2$ & $\bm{W}_2^2$ &$\cdots$ & $\bm{W}_N^2$  \\
 $\vdots$ & $\vdots$ & $\vdots$ & $\vdots$  \\
 $\bm{W}_1^M$ & $\bm{W}_2^M$ &$\cdots$ & $\bm{W}_N^M$ \\
\end{tabular}
\endgroup

\caption{Data stored in each server.}
\end{table}

Since the code is MDS, each $\bm{W}^j$ is determined by any $K$-subset of $\{ \bm{W}^j_1 , \ldots, \bm{W}^j_N \}$.

The data on server $i$ is $\bm{D}_i = (\bm{W}^1_i , \ldots, \bm{W}^M_i) \in \mathbb{F}_q^{ML/K}$.

A \emph{linear query} (from now on we omit the term linear) is a vector $\bm{q} \in \mathbb{F}_q^{ML/K}$. When a user sends a query $\bm{q}$ to a server $i$, this server answers back with the inner product $\langle \bm{D}_i , \bm{q} \rangle \in \mathbb{F}_q$.

The problem of private information retrieval can be stated informally as follows: A user wishes to download a file $\bm{W}^m$ without leaking any information about $m$ to any of the servers where at most $T$ of them may collude. The goal is for the user to achieve this while minimizing the download cost.

The messages $\bm{W}^1, \bm{W}^2, \ldots, \bm{W}^M$ are assumed to be independent and uniformly distributed elements of $\mathbb{F}_q^L$. The user is interested in a message $\bm{W}^m$. The index of this message, $m$, is chosen uniformly at random from the set $\{1,2,\ldots,M\}$.

A \emph{PIR scheme} is a set of queries for each possible desired message $\bm{W}^m$. We denote a scheme by $\mathcal{Q} = \{ \mathcal{Q}^1, \ldots, \mathcal{Q}^M \}$ where $\mathcal{Q}^m = \{ Q_1^m, \ldots, Q_N^m \}$ is the set of queries which the user will send to each server when they wish to retrieve $\bm{W}^m$. So, if the user is interested in $\bm{W}^m$, $Q^m_i$ denotes the set of queries sent to server $i$. The set of answers,  $\mathcal{A} = \{ \mathcal{A}^1, \ldots, \mathcal{A}^M \}$, is defined analogously.

A PIR scheme should satisfy two properties:
\begin{enumerate}
\item Correctness: $H(\bm{W}^m | \mathcal{A}^m) = 0$.
\item $T$-Privacy:  $I(\cup_{j \in J} \mathcal{Q}_j^m ; m) = 0$, for every $J \subseteq [M]$ such that $|J|=T$, where $[M]=\{1, \ldots, M \}$.
\end{enumerate}

Correctness guarantees that the user will be able to retrieve the message of interest. $T$-Privacy guarantees that no $T$ colluding servers will gain any information on the message in which the user is interested.

\begin{definition}
Let $M$ messages be stored using an $(N,K)$-MDS code on $N$ servers. An $(N,K,T,M)$-PIR scheme is a scheme which satisfies correctness and $T$-Privacy.
\end{definition}

Note that $T$-Privacy implies in  $|\mathcal{Q}^1|=|\mathcal{Q}^i|$ for every $i$, i.e., the number of queries does not depend on the desired message.

\begin{definition}
The \emph{PIR rate} of an $(N,K,T,M)$-PIR scheme $\mathcal{Q}$ is $R_\mathcal{Q}= \frac{L}{|\mathcal{Q}^1|}$.
\end{definition}

The PIR rate is the reciprocal of the download cost.

\section{One-Shot Schemes}

In this section, we introduce the notion of a one-shot scheme. The main idea behind one-shot schemes is encrypting the queries from the servers, as shown in Example \ref{ex: intro}. 

In the last two subsections, we give examples of known schemes from the literature, which we interpret as one-shot schemes, and present new ones.

\subsection{Definition and Examples}

Throughout this paper we assume, without loss of generality, that the user is interested in retrieving the first message. 

We denote the subspace of queries which only query the first message by $\mathbb{V}_1 = \{ \bm{a} \in \mathbb{F}_q^{ML/K}  : i>L/K \Rightarrow \bm{a}_i = 0 \}$.

\begin{definition}[One-Shot Schemes] \label{def: one-shot scheme}
An $(N,K,T,M)$-one-shot scheme of codimension $r$ is a scheme in which the queries have the following form.

\begin{table}[H]
\centering
\begin{tabular}{c c c c c c}
Server $1$ & $\cdots$ & Server $r$ & Server $r+1$ & $\cdots$ & Server $N$ \\
\toprule
$\bm{q}_1$ & $\cdots$ & $\bm{q}_r$ & $\bm{q}_{r+1}+\bm{a}_1 $ & $\cdots$ & $\bm{q}_N+\bm{a}_{N-r}$  \\ 
\end{tabular}
\caption{Query structure for a one-shot scheme.}
\label{table:oneshot}
\end{table}

The queries in Table \ref{table:oneshot} satisfy the following properties:
\begin{enumerate}
\item The $\bm{q}_1, \ldots, \bm{q}_N \in \mathbb{F}_q^{ML/K}$ and $\bm{a}_1, \ldots, \bm{a}_{N-r} \in \mathbb{V}_1$.
\item Any collection of $T$ queries sent to the servers is uniformly and independently distributed. \label{property1}
\item The $\bm{a}_1, \ldots, \bm{a}_{N-r} \in \mathbb{V}_1$ are such that the responses $\langle \bm{D}_{r+1} , \bm{a}_1 \rangle , \ldots, \langle \bm{D}_N , \bm{a}_{N-r} \rangle$ are linearly independent. \label{property2}
\item For every $i>r$, the response $\langle \bm{D}_i , \bm{q_i} \rangle$ is a linear combination of $\langle \bm{D}_1 , \bm{q}_1 \rangle , \ldots, \langle \bm{D}_r , \bm{q_r} \rangle$. \label{property3}
\end{enumerate}

We refer to queries of the form $\bm{q}_* + \bm{a}_*$ as mixed queries, where the noisy query $\bm{q}_*$ is hiding the informative query $\bm{a}_*$.
\end{definition}

Property \ref{property1} ensures $T$-privacy. Properties \ref{property2} and \ref{property3} ensure correctness. The PIR rate of a one-shot scheme is a function, solely, of its codimension and is given in Proposition \ref{pro:oneshotpir}. First, we need the following definition.

\begin{definition}[Decoding Equations]
Consider the one-shot scheme in Definition \ref{def: one-shot scheme}. Property \ref{property3} implies that, for every $i > r$, there exists an equation of the form
\[ \langle \bm{D}_i, \bm{q}_i \rangle = \alpha^i_1 \langle \bm{D}_1, \bm{q}_1 \rangle + \ldots + \alpha^i_r \langle \bm{D}_r, \bm{q}_r \rangle .\]
We call these equations, one for every $i>r$, the decoding equations of the one-shot scheme.
\end{definition}

\begin{proposition} \label{pro:oneshotpir}
A one-shot PIR scheme, $\mathcal{Q}$, of codimension $r$ has rate
\[R_\mathcal{Q} = \frac{N-r}{N} = 1 - \frac{r}{N} .\]
\end{proposition}
\begin{proof}
For every $i>r$, we have a decoding equation
\[ \langle \bm{D}_i, \bm{q}_i \rangle = \alpha^i_1 \langle \bm{D}_1, \bm{q}_1 \rangle + \ldots + \alpha^i_r \langle \bm{D}_r, \bm{q}_r \rangle .\]
We can then write, for every $1 \leq i \leq N-r$,
\[ \langle \bm{D}_{r+i}, \bm{a}_i \rangle = \langle \bm{D}_{r+i}, \bm{q}_{r+i} + \bm{a}_i \rangle - \langle \bm{D}_i, \bm{q}_i \rangle .\] 
Thus, the user can retrieve the linearly independent $\langle \bm{D}_{r+1} , \bm{a}_1 \rangle , \ldots, \langle \bm{D}_N , \bm{a}_{N-r} \rangle$, obtaining $N-r$ useful symbols, having downloaded a total of $N$ symbols. 
\end{proof}

\begin{remark} \label{rem: divisibility}
Technically,  $N-r$ must be divisible by $L$. When this does not occur, the one-shot scheme must be repeated  $\lcm(N-r,L)$ times\footnote{We denote the least common multiple of $N-r$ and $L$ by $\lcm(N-r,L)$.}. This, however, does not change the rate of the scheme. We explain more in detail in Example \ref{exe:oneshot coded}. 
\end{remark}

The scheme in Table \ref{table:intro} is an example of a one-shot scheme with codimension $1$. We now present a more complicated example, with codimension $2$, where the data is replicated.

\begin{example}[Replicated Data with Codimension $2$] \label{exe: one-shot 4,1,2,M}
Suppose $M$ messages, $\bm{W}^1, \ldots, \bm{W}^M \in \mathbb{F}_3^2 $ are stored, replicated, onto $N=4$ servers, where at most $T=2$ of them collude, and a user is interested in retrieving $\bm{W}^1$ privately. The user can use the following one-shot scheme.

\begin{table}[H]
\centering
\begin{tabular}{c c c c}
 Server $1$ &  Server $2$ & Server $3$ & Server $4$ \\
\toprule
$\bm{q}_1$ & $\bm{q}_2$ & $\bm{q}_1 + \bm{q}_2 + \bm{e}_1$ &  $\bm{q}_1 + 2\bm{q}_2 + \bm{e}_2$   \\
\end{tabular}
\caption{Query structure for Example \ref{exe: one-shot 4,1,2,M}.}
\label{table: one-shot 4,1,2,M}
\end{table}
The queries in Table \ref{table: one-shot 4,1,2,M} satisfy the following properties:
\begin{itemize}
\item The vectors $\bm{q}_1,\bm{q}_2 \in \mathbb{F}_3^{2M}$ are chosen uniformly at random.
\item The vectors $\bm{e}_1$ and $\bm{e}_2$ are the first two vectors of the standard basis of $\mathbb{F}_3^{2M}$.
\end{itemize}

This scheme is private since for any two servers the queries are uniformly and independently distributed.

The two decoding equations of this one-shot scheme are
\begin{align}
\langle \bm{D} , \bm{q}_1 + \bm{q}_2 \rangle &= \langle \bm{D} , \bm{q}_1 \rangle + \langle \bm{D} , \bm{q}_2 \rangle \label{eq:eq1} \\
\langle \bm{D} , \bm{q}_1 + 2\bm{q}_2 \rangle &= \langle \bm{D} , \bm{q}_1 \rangle + 2\langle \bm{D} , \bm{q}_2 \rangle . \label{eq:eq2}
\end{align}

From these equations and the responses from the server, the user can retrieve $\bm{W}^1_1 = \langle \bm{D} , \bm{e}_1 \rangle$ and $\bm{W}^1_2 = \langle \bm{D} , \bm{e}_2 \rangle$.

To retreive $2$ units of information from the message the user has to download $4$. Therefore, the rate of this PIR scheme is $R=1/2$, as per Proposition \ref{pro:oneshotpir}.

\end{example}

In the following example we will consider a coded database where we must take Remark \ref{rem: divisibility} into account.

\begin{example}[Coded Data] \label{exe:oneshot coded}
Suppose $M$ messages are stored using a $(4,2)$-MDS code over $\mathbb{F}_3$ as in Table \ref{table:example}.

\begin{table}[H]
\centering
\begingroup
\renewcommand{\arraystretch}{1.5}
\begin{tabular}{ c  c  c  c}
Server $1$ & Server $2$ & Server $3$ & Server $4$ \\
\toprule
 $\bm{W}_1^1$ & $\bm{W}_2^1$ &$\bm{W}_1^1+\bm{W}_2^1$ & $\bm{W}_1^1+2 \bm{W}_2^1$ \\
 $\bm{W}_1^2$ & $\bm{W}^2_2$ & $\bm{W}_1^2 + \bm{W}^2_2$ & $\bm{W}_1^2 + 2 \bm{W}^2_2$ \\
 $\vdots$ & $\vdots$ & $\vdots$ & $\vdots$  \\
 $\bm{W}_1^M$ & $\bm{W}^M_2$ & $\bm{W}_1^M + \bm{W}^M_2$ & $\bm{W}_1^M + 2 \bm{W}^M_2$
\end{tabular}
\endgroup

\caption{Data stored in a $(4,2)$-MDS code.}
\label{table:example}
\end{table}

Suppose the user is interested in the first message and wants $2$-privacy, i.e., at most $2$ servers can collude. The following is a $(4,2,2,M)$-one-shot scheme taken from \cite{tajeddine2016private}.
\begin{table}[H]
\centering
\begin{tabular}{c c c c}
Server $1$ &  Server $2$ & Server $3$ & Server $4$ \\
\toprule
$\bm{q}_1$ & $\bm{q}_2$ & $\bm{q}_3$ & $\bm{q}_4+\bm{e}_1$  
\end{tabular}
\caption{Query structure for Example \ref{exe:oneshot coded}.}
\label{table:oneshot coded}
\end{table}
The queries in Table \ref{table:oneshot coded} satisfy the following properties:
\begin{itemize}
\item The vectors $\bm{q}_1,\bm{q}_2 \in \mathbb{F}_3^{M}$ are uniformly and independently distributed.
\item The vectors $\bm{q}_3=\bm{q}_1+\bm{q}_2$ and $\bm{q}_4=\bm{q}_1+2\bm{q}_2$.
\item The vector $\bm{e}_1$ is the first vector of the standard basis of $\mathbb{F}_3^{M}$
\end{itemize}

This scheme is private since for any two servers the queries are uniformly and independently distributed.

The decoding equation of this one-shot scheme is 
\begin{align} \label{eq:oneshotidentity}
\langle \bm{D}_4 , \bm{q}_4 \rangle = - \langle \bm{D}_1 , \bm{q}_1 \rangle +2 \langle \bm{D}_2 , \bm{q}_2 \rangle + 2 \langle \bm{D}_3 , \bm{q}_3 \rangle .
\end{align}

From this equation the user can retrieve $\bm{W}^1_1+2\bm{W}^1_2$. To retrieve $1$ unit of the message the user has to download $4$ units. Therefore, the rate of the PIR scheme is $R=1/4$, which could have also been obtained from Proposition \ref{pro:oneshotpir}.

The scheme, however, does not satisfy correctness, since the user has not retrieved the whole of $\bm{W}^1$. To do this, we must repeat it ${\lcm(N-r,L) = 2}$ times, as stated in Remark \ref{rem: divisibility}.

To get a different combination we will have to change where we introduce our mixed query. Consider the following query structure, analagous to the one in Table \ref{exe:oneshot coded}

\begin{table}[H]
\centering
\begin{tabular}{c c c c}
Server $1$ &  Server $2$ & Server $3$ & Server $4$ \\
\toprule
$\bm{q}_5$ & $\bm{q}_6$ & $\bm{q}_7 +\bm{e}_1 $ & $\bm{q}_8$ 
\end{tabular}
\caption{Second part of the query structure in Example~\ref{exe:oneshot coded}.}
\label{table:oneshot coded 2}
\end{table}

The queries in Table \ref{table:oneshot coded} satisfy the following properties:
\begin{itemize}
\item The vectors $\bm{q}_5,\bm{q}_6 \in \mathbb{F}_3^{M}$ are uniformly and independently distributed.
\item The vectors $\bm{q}_7=\bm{q}_5+\bm{q}_6$ and $\bm{q}_8=\bm{q}_5+2\bm{q}_6$.
\item The vector $\bm{e}_1$ is the first vector of the standard basis of $\mathbb{F}_3^{M}$
\end{itemize}

The decoding equation of this one-shot scheme, the same as Equation \ref{eq:oneshotidentity}, can be written as
\begin{align} \label{eq:oneshotidentityb}
\langle \bm{D}_3 , \bm{q}_3 \rangle = 2 \langle \bm{D}_1 , \bm{q}_1 \rangle - \langle \bm{D}_2 , \bm{q}_2 \rangle + 2 \langle \bm{D}_4 , \bm{q}_4 \rangle .
\end{align}

From this equation the user can retrieve $\bm{W}^1_1+\bm{W}^1_2$. Now, having both $\bm{W}^1_1+2\bm{W}^1_2$ and $\bm{W}^1_1+\bm{W}^1_2$, the user can decode $W^1$. To retrieve $2$ units of information the user had to download $8$ units. Thus, as mentioned in Remark \ref{rem: divisibility}, the rate of the one-shot scheme is still $R=1/4$.

\end{example}

In the rest of the paper we focus only on the first round of queries since, as stated in Remark \ref{rem: divisibility}, the rate of the scheme is not affected by the additional rounds needed for correctness.

\begin{remark} \label{rem:fieldoneshot}
The schemes in Examples \ref{exe: one-shot 4,1,2,M} and \ref{exe:oneshot coded} do not work over the base field $\mathbb{F}_2$. The reason for this is that the decoding equations \ref{eq:eq1}, \ref{eq:eq2}, \ref{eq:oneshotidentity}, and \ref{eq:oneshotidentityb} all have $2$ as a coefficient. Since $\mathbb{F}_2$ has characteristic $2$, the coefficient $2$ is equal to $0$.\footnote{The characteristic of a field, $F_q$, is the smallest integer $k$ such that, for every $x \in \mathbb{F}_q$, $\underbrace{x+\ldots+x}_{\text{$k$ times}} = 0.$}

In general, given a one-shot scheme, the base field must have a characteristic which is compatible with the coefficients of the decoding equations. For simplicity, we can always consider that the base field has characteristic larger than the largest coefficient in the decoding equations. 
\end{remark}

\subsection{One-Shot Schemes from the Literature}

In this section, we reinterpret some schemes from the literature as one-shot schemes. As usual, $N$ denotes the number of servers, $T$ the number of colluding servers, and $M$ the number of messages. The messages are stored as an $(N,K)$-MDS code and the user wants the first message. \mbox{}\\

\noindent{\em Replicated Data:} When the messages are replicated among the servers, i.e. $K=1$, a $T$-threshold linear secret sharing scheme \cite{CDN15} can be transformed into the following one-shot scheme, with codimension $r=T$.

\begin{table}[H]

\begin{tabular}{c  c  c  c c c}
Server $1$ & \ldots & Server $T$ & Server $T+1$ & \ldots & Server $N$ \\
\toprule
$\bm{q}_1$ &  & $\bm{q}_T$ & $\bm{q}_{T+1}+\bm{e}_1 $ & & $\bm{q}_N + \bm{e}_{N-T}$ 
\end{tabular}

\caption{Query structure for the secret sharing scheme.}
\label{tab:secret sharing oneshot}
\end{table}

The queries in Table \ref{tab:secret sharing oneshot} satisfy the following properties:
\begin{itemize}
\item The vectors $\bm{q}_1, \ldots, \bm{q}_T \in \mathbb{F}_q^{M(N-T)}$ are chosen uniformly at random.
\item The vectors $\bm{q}_{T+1}, \ldots, \bm{q}_{N-T} \in \mathbb{F}_q^{M(N-T)}$ are such that $\bm{q}_1, \ldots, \bm{q}_N$ are a $(N,T)$-MDS.
\item The vectors $\bm{e}_1, \ldots, \bm{e}_{N-T}$ are the first $N-T$ vectors of the standard basis of $\mathbb{F}_q^{M(N-T)}$.
\end{itemize}

The schemes in Examples \ref{ex: intro} and \ref{exe: one-shot 4,1,2,M} are secret sharing schemes. The rate of these schemes is given by $R = \frac{N-T}{N}$, as per Proposition \ref{pro:oneshotpir}, and is known \cite{sun2016capacitynoncol} to be asymptotically optimal when the number of messages, $M$, goes to infinity. \mbox{}\\

\noindent{\em Coded Data:} The first one-shot scheme for coded databases was presented in \cite{tajeddine2016private}. The scheme in Theorem $3$ of the journal version \cite{taje18}, in our notation, has codimension $r=\frac{NK-N+T}{K}$.

An improved PIR schemes for MDS coded data with collusions was presented in \cite{freij2016private} which, in our notation, has codimension $r=K+T-1$.

It is easy to compare the rate of these schemes via their codimension. The scheme in \cite{freij2016private} outperforms the scheme in \cite{taje18} unless in the cases where either $K=1$ or $N=K+T$. 

\subsection{Geometrical One-Shot Schemes} \label{sec:new oneshot}

In this section we will present a new family of one-shot scheme with simple constructions which we refer to as geometrical one-shot schemes. They exist for the parameters $K=2$ or $N=K+T$, and have the same codimension as the schemes in \cite{freij2016private} but with a more elementary construction, not requiring the use of the star product.

\begin{construction}
Let $M$ messages be stored on $N$ servers, using a $(N,K)$-MDS code, where at most $T$ servers collude. 

Since the data is stored as an $(N,K)$-MDS code, there exists unique coefficients $\lambda^i_j$ such that, for every $i \in [T]$, and for every $l \in [K+T,N]$,
\[\bm{D}_i = \sum_{j=T+1}^{K+T-1} \lambda^i_j(l) \bm{D}_j + \lambda^i_l(l) \bm{D}_l , \]
where $\bm{D}_i$ is the data stored in Server $i$.

Suppose that $\frac{\lambda_j^i (l_0)}{\lambda_j^i (l)} = \gamma_i^l$, for every $l \in [K+T,N]$, where $l_0 = K+T$.

We present the following one-shot scheme.

\begin{table}[H]
\centering
\resizebox{!}{0.35cm}{
\begin{tabular}{c c c c c c }
Server $1$ & \ldots & Server $K+T-1$ & Server $K+T$ & \ldots & Server $N$ \\
\toprule
$\bm{q}_1$ &  & $\bm{q}_{K+T-1}$ & $\bm{q}_{K+T}+\bm{e}_1$ & & $\bm{q}_{N}+\bm{e}_{N-K-T+1}$
\end{tabular}
}

\caption{Query structure for Scheme 1.}
\label{tab:scheme 1}
\end{table}

The queries in Table \ref{tab:scheme 1} satisfy the following properties:
\begin{itemize}
\item The vectors $\bm{q}_1, \ldots, \bm{q}_T \in \mathbb{F}_q^{LM/K}$ are chosen uniformly at random.
\item The vectors $\bm{q}_{T+1}, \ldots, \bm{q}_{N} \in \mathbb{F}_q^{LM/K}$ are such that ${\bm{q}_j =  \sum_{i=1}^T \lambda_{j}^i(l_0) \bm{q}_i}$, for every $j \in [T+1,N]$.
\item The vectors $\bm{e}_1, \ldots \bm{e}_{N-(K+T-1)}$ are the first $N-(K+T-1)$ vectors of the standard basis of $\mathbb{F}_q^{LM/K}$.
\end{itemize}

\end{construction}

\begin{proposition}
Geometrical one-shot schemes satisfy correctness and $T$-privacy and are, therefore, one-shot schemes.
\end{proposition}

\begin{proof}

For every $l \in [K+T,N]$,

\begin{align*}
    &\sum_{i=1}^T \gamma_i^l \langle \bm{D}_i , \bm{q}_i \rangle = \sum_{i=1}^T \gamma_i^j \langle \sum_{j=T+1}^{K+T-1} \lambda^i_j(l) \bm{D}_j + \lambda^i_l(l) \bm{D}_l , \bm{q}_i \rangle \\
    &= \sum_{i=1}^T \left( \sum_{j=T+1}^{K+T-1} \gamma_i^l \lambda_j^i (l) \langle \bm{D}_j , \bm{q}_i + \gamma_i^l \lambda_l^i (l) \langle \bm{D}_l , \bm{q}_i \rangle \right) \\
    &= \sum_{j=T+1}^{K+T-1} \langle \bm{D}_j , \sum_{i=1}^T \gamma_i^l \lambda_j^i (l) \bm{q}_i \rangle + \langle \bm{D}_l , \sum_{i=1}^T \gamma_i^l \lambda_l^i (l) \bm{q}_i \rangle \\
    &= \sum_{j=T+1}^{K+T-1} \langle \bm{D}_j , \bm{q}_j \rangle + \langle \bm{D}_l , \bm{q}_l \rangle 
\end{align*}

So that we get the decoding equations

\begin{align*}
   \langle \bm{D}_l , \bm{q}_l \rangle  = \sum_{i=1}^T \gamma_i^l \langle \bm{D}_i , \bm{q}_i \rangle - \sum_{j=T+1}^{K+T-1} \langle \bm{D}_j , \bm{q}_j \rangle
\end{align*}
for every $l \in [K+T,N]$.
\end{proof}

In the next example we construct a geometrical one-shot scheme for the setting of Example \ref{exe:oneshot coded}.

\begin{example} \label{exe:geometric oneshot coded}

Consider the setting of Example \ref{exe:oneshot coded}. Via the geometrical one-shot scheme described in this subsection we obtain the following one-shot scheme.

\begin{table}[H]
\centering
\begin{tabular}{c c c c}
Server $1$ &  Server $2$ & Server $3$ & Server $4$ \\
\toprule
$\bm{q}_1$ & $\bm{q}_2$ & $\bm{q}_3$ & $\bm{q}_4+\bm{e}_1$  
\end{tabular}
\caption{Query structure for Example \ref{exe:geometric oneshot coded}.}
\label{table:geometric oneshot coded}
\end{table}
The queries in Table \ref{table:geometric oneshot coded} satisfy the following properties:
\begin{itemize}
\item The vectors $\bm{q}_1,\bm{q}_2 \in \mathbb{F}_3^{M}$ are uniformly and independently distributed.
\item The vectors $\bm{q}_3= 2 \bm{q}_1 - \bm{q}_2$ and $\bm{q}_4=- \bm{q}_1+ \bm{q}_2$.
\item The vector $\bm{e}_1$ is the first vector of the standard basis of $\mathbb{F}_3^{M}$
\end{itemize}

The decoding equation of this one-shot scheme is 
\begin{align} 
\langle \bm{D}_4 , \bm{q}_4 \rangle =  \langle \bm{D}_1 , \bm{q}_1 \rangle + \langle \bm{D}_2 , \bm{q}_2 \rangle - \langle \bm{D}_3 , \bm{q}_3 \rangle .
\end{align}

\end{example}

\section{The Refinement Lemma} \label{sec:refine}

In this section we present the first part of our refine and lift operation. The refinement lemma shows how to improve the rate of a one-shot scheme when the number of messages, $M$, is equal to $2$.

One of the ideas behind the refinement lemma is to query the messages individually. For this we will need to consider the following linear subspaces.

\begin{definition}
We denote by 
\begin{equation*}
\resizebox{.95\hsize}{!}{$\mathbb{V}_j = \{ \bm{b} \in \mathbb{F}_q^{ML/K}  : \text{$i < (j-1)L/K+1$ or $i>jL/K$} \Rightarrow \bm{b}_i = 0 \}$,}
\end{equation*}
the subspace of queries which only query the $j$-th message.
\end{definition}

\begin{example}[Refined One-Time Pad Scheme] \label{exe: intro refinement}
Suppose $M=2$ two-bit messages, $\bm{W}^1, \bm{W}^2 \in \mathbb{F}_2^2$, are stored, replicated, onto two non colluding servers and a user is interested in retrieving $\bm{W}^1$ privately.

This setting could be solved by applying the one-shot scheme in Example \ref{ex: intro} with a PIR rate of $r = 1/2$.

Consider, however, the following scheme.

\begin{table}[H]
\centering
\begin{tabular}{c c }
Server $1$ &  Server $2$ \\
\toprule
$\bm{a}_1$ & $\bm{b_1}+\bm{a}_2$   \\
$\bm{b}_1$ & 
\end{tabular}
\caption{Query structure for Example \ref{exe: intro refinement}.}
\label{table:intro refinement}
\end{table}

The queries in Table \ref{table:intro refinement} satisfy the following properties:
\begin{itemize}
\item The vector $\bm{b}_1 \in \mathbb{V}_2$ is chosen uniformly at random and such that $\langle \bm{D} , \bm{b} \rangle$ is linearly independent, which in this case means that $\bm{b} \neq 0$.
\item The vectors $\bm{a}_1, \bm{a}_2 \in \mathbb{V}_1$ are chosen uniformly at random and such that $\langle \bm{D} , \bm{a}_1 \rangle$ and $\langle \bm{D} , \bm{a}_2 \rangle$ are linearly independent, which in this case means that $\bm{a}_1$ and $\bm{a}_2$ are linearly independent.
\end{itemize}

The scheme is private since, from any of the servers' point, the $\bm{a}$'s and $\bm{b}$'s are statistically indistinguishable.

To decode $\bm{W}^1$ the user uses the following decoding equation
\begin{align} \label{eq: refinement intro decoding}
\langle \bm{D} , \bm{b}_1 \rangle + \langle \bm{D} , \bm{b}_1+\bm{a}_2 \rangle = \langle \bm{D} , \bm{a}_2 \rangle
\end{align}
It is no coincidence that Equation \ref{eq: refinement intro decoding} has the same form as Equation \ref{equ:intro decoding}. Indeed, the role of $\bm{b}_1$ is analogous to that of $\bm{q}$ in Example \ref{exe: intro refinement}.

Having, the linearly indepenedent, $\langle \bm{D} , \bm{a}_1 \rangle$ and $\langle \bm{D} , \bm{a}_2 \rangle$, the user is able to decode $\bm{W}^1$. By downloading a total of $3$ bits, the user is able to privately retrieve the two bits of $\bm{W}^1$. Thus, the PIR rate of this scheme is $R=2/3$. This rate actually achieves the capacity for these parameters, given in Equation~\ref{eq:jafarcapacity}.
\end{example}

\begin{lemma}[The Refinement Lemma] \label{lem:ref}
Let $\mathcal{Q}$ be a one-shot scheme of co-dimension $r$, with rate $\frac{N-r}{N}$. Then, there exists  an $(N,K,T,2)$-PIR scheme, $\mathcal{Q'}$, with rate $R_{\mathcal{Q'}} = \frac{N}{N+r} > \frac{N-r}{N}$.
\end{lemma}
\begin{proof}
We construct $\mathcal{Q}'$ in the following way.
\begin{table}[H]
\centering
\begin{tabular}{c c c c c c}
Server $1$ & $\cdots$ & Server $r$ & Server $r+1$ & $\cdots$ & Server $N$ \\
\toprule
 $\bm{a}_1$ & $\cdots$ & $\bm{a}_r$ & $\bm{a}_{r+1}+\bm{b}_{r+1} $ & $\cdots$ & $\bm{a}_N+\bm{b}_N$  \\
 $\bm{b}_1$ & $\cdots$ & $\bm{b}_r$ & & &
\end{tabular}
\caption{Query structure for the refined scheme $\mathcal{Q}'$.}
\label{table:refinementlemma}
\end{table}
The queries in Table \ref{table:refinementlemma} satisfy the following properties:
\begin{itemize}
\item The vectors $\bm{b}_1, \ldots, \bm{b}_N \in \mathbb{V}_2$ are chosen with the same distribution as  $\bm{q}_1, \ldots, \bm{q}_N \in \mathbb{F}_q^{ML/K}$ from the one-shot scheme $\mathcal{Q}$ but such that  $\langle \bm{D}_1 , \bm{b}_1 \rangle , \ldots , \langle \bm{D}_T , \bm{b}_T \rangle$ are linearly independent.

\item The vectors $\bm{a}_1, \ldots, \bm{a}_N \in \mathbb{V}_1$ are chosen uniformly at random but such that $\langle \bm{D}_1 , \bm{a}_1 \rangle , \ldots , \langle \bm{D}_N , \bm{a}_N \rangle$ are linearly independent.
\end{itemize}

The scheme is private since for any $T$ servers the $\bm{a}$'s and the $\bm{b}$'s are indistinguishable.

The decoding equations of the one-shot scheme are valid also for the $\bm{b}$'s. To retrieve $N$ units of information the user had to download $N+r$. Thus, the rate of the scheme is
\[R_{\mathcal{Q'}} = \frac{N}{N+r} > \frac{N-r}{N} = R_{\mathcal{Q}} .\]

\end{proof}

\begin{remark}
Since the refined scheme inherits the decoding equations from the one-shot scheme, they will both share the same restrictions on the base field, as discussed in Remark~\ref{rem:fieldoneshot}.

The larger the characteristic of the base field is, the larger the probability that $\langle \bm{D}_1 , \bm{a}_1 \rangle , \ldots , \langle \bm{D}_N , \bm{a}_N \rangle$ will be linearly independent if the vectors $\bm{a}_1, \ldots, \bm{a}_N \in \mathbb{V}_1$ are chosen uniformly at random.
\end{remark}

\begin{example}[Refined Scheme on Replicated Data] \label{exe: refined one-shot 4,1,2,M}
Refining the scheme in Example \ref{exe: one-shot 4,1,2,M} we get the following scheme.

\begin{table}[H]
\centering
\begin{tabular}{c c c c}
 Server $1$ &  Server $2$ & Server $3$ & Server $4$ \\
\toprule
$\bm{a}_1$ & $\bm{a}_2$ & $\bm{b}_1 + \bm{b}_2 + \bm{a}_3$ &  $\bm{b}_1 + 2\bm{b}_2 + \bm{a}_4$   \\
$\bm{b}_1$ & $\bm{b}_2$ & &
\end{tabular}
\caption{Query structure for Example \ref{exe: refined one-shot 4,1,2,M}.}
\label{table: exe: refined one-shot 4,1,2,M}
\end{table}
The queries in Table \ref{table: exe: refined one-shot 4,1,2,M} satisfy the following properties:
\begin{itemize}
\item The vectors $\bm{b}_1, \bm{b}_2 \in \mathbb{V}_2$ are chosen uniformly at random and such that $\langle \bm{D} , \bm{b}_1 \rangle$ and $\langle \bm{D} , \bm{b}_2 \rangle$ are linearly independent, which in this case means that $\bm{b}_1$ and $\bm{b}_2$ are linearly independent.

\item The vectors $\bm{a}_1, \ldots, \bm{a}_4 \in \mathbb{V}_1$ are chosen uniformly at random and such that $\langle \bm{D} , \bm{a}_1 \rangle, \ldots, \langle \bm{D} , \bm{a}_4 \rangle$ are linearly independent, which in this case means that $\bm{a}_1, \ldots, \bm{a}_4$ are linearly independent.
\end{itemize}

This scheme is private since for any two servers the $b$'s and $a$'s are indistinguishable.

The two decoding equations of this one-shot scheme are
\begin{align}
\langle \bm{D} , \bm{b}_1 + \bm{b}_2 \rangle &= \langle \bm{D} , \bm{b}_1 \rangle + \langle \bm{D} , \bm{b}_2 \rangle \label{eq:eq1 refine} \\
\langle \bm{D} , \bm{b}_1 + 2\bm{b}_2 \rangle &= \langle \bm{D} , \bm{b}_1 \rangle + 2\langle \bm{D} , \bm{b}_2 \rangle . \label{eq:eq2 refine}
\end{align}
These equations are inherited from the one-shot scheme, Equations \ref{eq:eq1 refine} and \ref{eq:eq2 refine} from Example \ref{exe: one-shot 4,1,2,M}.

From these equations, the user can retrieve $\langle \bm{D} , \bm{a}_3 \rangle$ and $\langle \bm{D} , \bm{a}_3 \rangle$, which together with $\langle \bm{D} , \bm{a}_1 \rangle$ and $\langle \bm{D} , \bm{a}_2 \rangle$ gives a total of $4$ units of information retrieved from a total of $6$ downloaded. Thus, the PIR rate of this scheme is $R = 2/3$, as per Lemma \ref{lem:ref}.

\end{example}

\begin{example}[Refined Scheme on Coded Data] \label{exe: refined oneshot coded}

Refining the scheme in Example \ref{exe:oneshot coded} we get the following scheme.

\begin{table}[H]
\centering
\begin{tabular}{c c c c}
Server $1$ &  Server $2$ & Server $3$ & Server $4$ \\
\toprule
$\bm{a}_1$ & $\bm{a}_2$ & $\bm{a}_3$ & $\bm{a}_4+\bm{b}_4$  \\
$\bm{b}_1$ & $\bm{b}_2$ & $\bm{b}_3$ &
\end{tabular}
\caption{Query structure for Example \ref{exe: refined oneshot coded}.}
\label{table:exe: refined oneshot coded}
\end{table}
The queries in Table \ref{table:exe: refined oneshot coded} satisfy the following properties:
\begin{itemize}
\item The vectors $\bm{b}_1, \bm{b}_2 \in \mathbb{V}_2$ are chosen uniformly at random and such that $\langle \bm{D}_1 , \bm{b}_1 \rangle$ and $\langle \bm{D}_2 , \bm{b}_2 \rangle$ are linearly independent.

\item The vectors $\bm{b}_3=\bm{b}_1+\bm{b}_2$ and $\bm{b}_4=\bm{b}_1+2\bm{b}_2$.

\item The vectors $\bm{a}_1, \ldots, \bm{a}_4 \in \mathbb{V}_1$ are chosen uniformly at random and such that $\langle \bm{D}_1 , \bm{a}_1 \rangle, \ldots, \langle \bm{D}_4 , \bm{a}_4 \rangle$ are linearly independent.
\end{itemize}

This scheme is private since for any two servers the $\bm{b}$'s and $\bm{a}$'s are indistinguishable.

The decoding equation of this one-shot scheme is
\begin{align} \label{eq: decoding refined oneshot coded}
\langle \bm{D}_4 , \bm{b}_4 \rangle = - \langle \bm{D}_1 , \bm{b}_1 \rangle +2 \langle \bm{D}_2 , \bm{b}_2 \rangle + 2 \langle \bm{D}_3 , \bm{b}_3 \rangle .
\end{align}
This equation is inherited from the one-shot scheme, Equation~\ref{eq:oneshotidentity} in Example \ref{exe:oneshot coded}.

From this equation, the user can retrieve $\langle \bm{D}_4 , \bm{a}_4 \rangle$, which together with $\langle \bm{D}_1 , \bm{a}_1 \rangle$, $\langle \bm{D}_2 , \bm{a}_2 \rangle$, and $\langle \bm{D}_3 , \bm{a}_3 \rangle$ gives a total of $4$ units of information retrieved from a total of $7$ downloaded. Thus, the PIR rate of this scheme is $R = 4/7$, as per Lemma \ref{lem:ref}.

\end{example}

\section{The Lifting Theorem} \label{sec:lift}

In this section, we present our main result in Theorem \ref{teo:lift}. We show how to extend, by means of a lifting operation, the refined scheme on two messages to any number of messages.

Informally, the lifting operation consists of two steps: a symmetrization step, and a way of dealing with ``leftover'' queries that result from the symmetrization. 

We also introduce a symbolic matrix representation for PIR schemes which simplifies our analysis. 

\subsection{An Example of the Lifting Operation} \label{sec:exlift}

\begin{definition}
A $k$-query is a sum of $k$ queries, each belonging to a different $\mathbb{V}_j$, $j \in [M]$.
\end{definition}

So, for example, if $\bm{a}\in \mathbb{V}_1$, $\bm{b}\in \mathbb{V}_2$, and $\bm{c}\in \mathbb{V}_3$, then $\bm{a}$ is a $1$-query, $\bm{a}+\bm{b}$ is a $2$-query, and $\bm{a}+\bm{b}+\bm{c}$ is a $3$-query.

Consider the scheme in Example \ref{exe: refined oneshot coded}. We represent the structure of this scheme, given in Table \ref{table:exe: refined oneshot coded},  by means of the following matrix:
\begin{align} \label{eq:s2}
 S_2 = \begin{pmatrix}
1 & 1 & 1 & 2
\end{pmatrix}
.\end{align}

Each column of $S_2$ corresponds to a server. A $1$ in column~ $i$ represents sending all possible combinations of $1$-queries of every message to server $i$, and a $2$ represents sending all combinations of $2$-queries of every message to server $i$. We call this matrix the \emph{symbolic matrix} of the scheme.

The co-dimension $r=3$ tells us that for every $r=3$ ones there is $N-r=1$ twos in the symbolic matrix.

Given the interpretation above, the symbolic matrix $S_2$ can be readily applied to obtain the structure of a PIR scheme for any number of messages $M$. For $M=3$, the structure is as follows.
\begin{table}[H]
\centering
\begin{tabular}{c  c  c  c}
server $1$ & server $2$ & server $3$ & server $4$ \\
\toprule
$\bm{a}_1$ & $\bm{a}_2$ & $\bm{a}_3$ & $\bm{a}_4 + \bm{b}_4 $  \\
$\bm{b}_1$ & $\bm{b}_2$ & $\bm{b}_3$ &  $\bm{a}_5 + \bm{c}_4 $\\
$\bm{c}_1$ & $\bm{c}_2$ & $\bm{c}_3$ &  $\bm{b}_5 + \bm{c}_5$\\
\end{tabular}
\caption{Query structure for $M=3$ in Example \ref{table:exe: refined oneshot coded} as implied by the symbolic matrix in Equation \ref{eq:s2}.}
\label{table:partiallift}
\end{table}

The queries in Table \ref{table:partiallift} satisfy the following properties:

\begin{itemize}
\item The $\bm{a}_i \in \mathbb{V}_1$, $\bm{b}_i \in \mathbb{V}_2$, and $\bm{c}_i \in \mathbb{V}_3$.
\item The $\bm{a}$'s and $\bm{b}$'s are chosen as in Example \ref{exe: refined oneshot coded}.
\item The $\bm{c}$'s are chosen analogously to the $\bm{b}$'s.
\item The extra ``leftover" term $\bm{b}_5+\bm{c}_5$ is chosen uniformly at random but different than zero.
\end{itemize}

The scheme in Table \ref{table:partiallift} has rate $5/12$. In this scheme, the role of $\bm{b}_5+\bm{c}_5$ is to achieve privacy and does not contribute to the decoding process. In this sense, it can be seen as a ``leftover'' query of the symmetrization. By repeating the scheme $r=3$ times, each one shifted to the left, so that the ``leftover'' queries appear in different servers, we can apply the same idea in the one-shot scheme to the ``leftover" queries, as shown in Table~\ref{table:lifted}. Thus, we   improve the rate from $5/12$ to $16/37$.

\begin{table}[H]
\centering
\begin{tabular}{c  c  c  c}
server $1$ & server $2$ & server $3$ & server $4$ \\
\toprule
$\bm{a}_1$ & $\bm{a}_2$ & $\bm{a}_3$ & $\bm{a}_4 + \bm{b}_4 $  \\
$\bm{b}_1$ & $\bm{b}_2$ & $\bm{b}_3$ &  $\bm{a}_5 + \bm{c}_4 $\\
$\bm{c}_1$ & $\bm{c}_2$ & $\bm{c}_3$ &  $\bm{b}_5 + \bm{c}_5$\\ \hdashline[2pt/1pt]
$\bm{a}_7$ & $\bm{a}_8$ & $\bm{a}_9 + \bm{b}_9 $  & $\bm{a}_6$  \\
$\bm{b}_7$ & $\bm{b}_8$ & $\bm{a}_{10} + \bm{c}_9 $  &  $\bm{b}_6$ \\
$\bm{c}_7$ & $\bm{c}_8$ & $\bm{b}_{10} + \bm{c}_{10}$  & $\bm{c}_6$ \\
\hdashline[2pt/1pt]
$\bm{a}_{13}$ & $\bm{a}_{14} + \bm{b}_{14} $ & $\bm{a}_{11}$ & $\bm{a}_{12}$   \\
$\bm{b}_{13}$ & $\bm{a}_{15} + \bm{c}_{14} $ & $\bm{b}_{11}$ & $\bm{b}_{12}$ \\
$\bm{c}_{13}$ & $\bm{b}_{15} + \bm{c}_{15}$ & $\bm{c}_{11}$ & $\bm{c}_{12}$ \\
\hdashline[2pt/1pt]
$\bm{a}_{16}+\bm{b}_{16}+\bm{c}_{16}$ &  &  &   \\
\end{tabular}
\caption{Query structure for lifted version of Table \ref{table:partiallift}.}
\label{table:lifted}
\end{table}

The queries in Table \ref{table:lifted} satisfy the following properties:

\begin{itemize}
\item The scheme is separated into four rounds.
\item In each of the first three rounds the queries behave as in Table \ref{table:partiallift}, but shifted to the left so that the ``leftover'' queries appear in different servers.
\item But now, $\bm{b}_{16}+\bm{c}_{16}$, $\bm{b}_{15}+\bm{c}_{15}$, $\bm{b}_{10}+\bm{c}_{10}$, and $\bm{b_{5}}+\bm{c_{5}}$ are chosen analogously to  $\bm{b}$'s in Example \ref{exe: refined oneshot coded}.
\end{itemize}

More precisely, the $2$-queries are chosen as follows.

\begin{itemize}
    \item The vectors $\bm{b}_{16}+\bm{c}_{16}, \bm{b}_{15}+\bm{c}_{15} \in \mathbb{V}_2 + \mathbb{V}_3$ are chosen uniformly at random and such that $\langle \bm{D}_1 , \bm{b}_{16}+\bm{c}_{16} \rangle$ and $\langle \bm{D}_2 , \bm{b}_{15}+\bm{c}_{15} \rangle$ are linearly independent.

\item The vectors $\bm{b}_{10}+\bm{c}_{10}= (\bm{b}_{16}+\bm{c}_{16}) + (\bm{b}_{15}+\bm{c}_{15})$ and $\bm{b_{5}}+\bm{c_{5}} = (\bm{b}_{16}+\bm{c}_{16} )+ 2 (\bm{b}_{15}+\bm{c}_{15})$.
\end{itemize}

In this way, $\langle \bm{D}_1, \bm{a}_{16} \rangle$ can be retrieved using the same decoding equation as Equation \ref{eq: decoding refined oneshot coded}.
\begin{multline*}\
\langle \bm{D}_1, \bm{b}_{16}+\bm{c}_{16} \rangle = 2 \langle \bm{D}_2, \bm{b}_{15}+\bm{c}_{15} \rangle  + 2 \langle \bm{D}_3, \bm{b}_{10}+\bm{c}_{10} \rangle \\ - \langle \bm{D}_1, \bm{b}_{5}+\bm{c}_{5} \rangle 
\end{multline*}

This scheme can be represented by the following matrix\footnote{We omit zeros in our symbolic matrices. }.
\[ S_3 = \begin{pmatrix}
1 & 1 & 1 & 2 \\
1 & 1 & 2 & 1 \\
1 & 2 & 1 & 1 \\
3 & & &
\end{pmatrix} \]

The scheme for $M=3$ messages was constructed recursively using the one for $2$ messages. It is this recursive operation that we call lifting. The main idea behind the lifting operation is that $r=3$ entries with value $k$ generate $N-r=1$ entry with value $k+1$ in the symbolic matrix.

Lifting $S_3$ to $S_4$ follows the same procedure: repeat $S_3$ $r=3$ times, each one shifted to the left, to produce $N-r=1$ $4$-query. As a result, we obtain the following symbolic matrix.

\[ S_4 = \begin{pmatrix}
1 & 1 & 1 & 2 \\
1 & 1 & 2 & 1 \\
1 & 2 & 1 & 1 \\
3 & & & \\
1 & 1 & 2 & 1 \\
1 & 2 & 1 & 1 \\
2 & 1 & 1 & 1 \\
 & & & 3 \\
1& 2 & 1 & 1 \\
2 & 1& 1 & 1 \\
1 & 1 & 1 & 2 \\
 & & 3 & \\
  & 4 &  & 
\end{pmatrix} \]

The queries are to be chosen analogously to the previous examples which we describe rigorously in the next subsection.

\subsection{The Symbolic Matrix}

In this section we define the symbolic matrix precisely. To do this, we will give a series of preliminary definitions. 

In what follows, we denote by $\mathbb{N}$ the set of non-negative integers, and by $\mathcal{M} (\mathbb{N})$ the set of matrices with entries in $\mathbb{N}$. If $S \in \mathcal{M} (\mathbb{N})$, we sometimes denote $S_{ij}$ by $S[i,j]$.

\begin{definition}
Let $S \in \mathcal{M} (\mathbb{N})$ be a matrix with $N$ columns. The shift operation is defined as the matrix $\sigma (S)$ such that 
\[ \sigma (S)[i,j] = \left\{\begin{array}{ll}
S[i,j+1] & \text{if $1\leq j \leq n-1$}, \\ 
S[i,1] & \text{if $j=n$}.
\end{array}\right. \]
\end{definition}

\begin{example}
If $S = \begin{pmatrix}
1 & 2 & 3\\ 
4 & 5 & 6
\end{pmatrix}$, then, $\sigma(S)=\begin{pmatrix}
2 & 3 & 1\\ 
5 & 6 & 4
\end{pmatrix}$.
\end{example}

\begin{definition}
Let $S \in \mathcal{M} (\mathbb{N})$ be a matrix. The set of entries in $S$ with value $k$ is denoted by $[k,S] = \{(i,j) : S[i,j] = k \}$.
\end{definition}

\begin{example}
If $S = \begin{pmatrix}
1 & 2 & 2 & 1 & 3 \\ 
4 & 2 & 3 & 3 & 2 
\end{pmatrix}$, then, $\# [2,S] = 4$, where ${[2,S] = \{(1,2), (1,3), (2,2), (2,5)\}}$.
\end{example}

\begin{definition}
We denote by $\prec$ the order on the set $\mathbb{N}^2$ given by $(i,j) \prec (i',j')$ if either $i<i'$ or $i=i'$ and $j<j'$.\footnote{This is known in the literature as the lexicographical order.}
\end{definition}

\begin{example}
If $x=(1,2)$, $y=(1,3)$, and $z=(2,1)$, then $x \prec y \prec z$.
\end{example}

\begin{definition}
Let $k \in \mathbb{N}$. We denote by $\tau_k : \mathbb{N}^2 \rightarrow  \mathbb{N}^2$ the function such that 
\[ \tau_k (i,j) = \left\{\begin{array}{ll}
(i+k,j-1) & \text{if $2\leq j \leq n-1$}, \\ 
(i+k,N) & \text{if $j=1$}.
\end{array}\right. \]
\end{definition}

\begin{example}
It holds that $\tau_2 (2,1) = (4,3)$.
\end{example}

\begin{definition}
Let $B \subseteq \mathbb{N}^2$. We denote by $\pi: 2^{\mathbb{N}^2} \rightarrow 2^{\mathbb{N}}$ the function such that $\pi (B) = \{ j \in \mathbb{N} : (i,j) \in B \}$.
\end{definition}

\begin{example}
If $B = \{ (1,3) , (2,1) \}$, then, $\pi (B) = \{1,3 \}$.
\end{example}

We are ready to define the symbolic matrix.

\begin{definition}[Symbolic Matrices] \label{def: symbolic matrix}
Let $N,M,r \in \mathbb{N}$. The symbolic matrix $S(N,M,r)$, which we will denote simply by $S_M$, is defined recursively as follows. 
\begin{align} 
S_2 &= (\overbrace{1,\ldots,1}^{r}, \overbrace{2,\ldots,2}^{N-r}) \label{eq:S_2} \\
S_{M+1} &= \lift(S_M) \label{eq:symbolicmatrix}
\end{align}
The function $\lift : \mathcal{M}(\mathbb{N}) \rightarrow \mathcal{M}(\mathbb{N})$ is given by

\[ \lift(S_M) = \begin{pmatrix}
S_M \\
\sigma (S_M) \\
\vdots \\
\sigma^{r-1} (S_M)\\
A
\end{pmatrix}, \]
where the matrix $A$ is constructed as follows.

Let $B = [M,S_M] = \{ b_1, \ldots, b_{ \# [M,S_M]} \} $, where $b_i \prec b_j$ if $i < j$. 

For every $i \in \mathbb{N}$, such that $1 \leq i \leq \# [M,S_M]$, we define the set $B_i = \{ b_i , \tau_l (b_i), \ldots, \tau_l^{r-1} (b_i) \}$, where $l$ is the number of rows in $S_M$.

Then, the matrix $A$ has $\# [M,S_M]$ rows and $N$ columns and is given by 
\[ A_{ij} = \left\{\begin{array}{ll}
0 & \text{if $j \in \pi (B_i)$}, \\ 
M+1 & \text{if $j \notin \pi (B_i)$}.
\end{array}\right. \]

\end{definition}

\begin{example}[Symbolic Matrix for Coded Data]
Consider the construction in Section \ref{sec:exlift}. In this case $N=4$ and $r=3$. The construction of $S_4$ in terms of $S_3$ takes the following form.

\begin{center}
\begin{tikzpicture}[decoration=brace]
   $S_4 =$ \matrix (m) [matrix of math nodes,left delimiter=(,right delimiter={)},xshift=4em] {
1 & 1 & 1 & 2 \\
1 & 1 & 2 & 1 \\
1 & 2 & 1 & 1 \\
3 & & & \\
1 & 1 & 2 & 1 \\
1 & 2 & 1 & 1 \\
2 & 1 & 1 & 1 \\
 & & & 3 \\
1& 2 & 1 & 1 \\
2 & 1& 1 & 1 \\
1 & 1 & 1 & 2 \\
 & & 3 & \\
  & 4 &  & \\
    };
    \draw[decorate,transform canvas={xshift=6.5em, yshift=0.4em},thick] (m-1-1.north west) -- node[right=2pt] {$S_3$} (m-5-1.north west);
    \draw[decorate,transform canvas={xshift=6.5em, yshift=0.2em},thick] (m-5-1.north west) -- node[right=2pt] {$\sigma(S_3)$} (m-9-1.north west);
        \draw[decorate,transform canvas={xshift=5.3em, yshift=0.1em},thick] (m-9-2.north west) -- node[right=2pt] {$\sigma^2(S_3)$} (m-13-2.north west);
        \draw[decorate,transform canvas={xshift=5.3em, yshift=-5.3em},thick] (m-9-2.north west) -- node[right=2pt] {$A$} (m-10-2.north west);
    
\end{tikzpicture}
\end{center}
\end{example}

In the next example the matrix $A$ takes a more complicated form.

\begin{example}[Symbolic Matrix for Replicated Data] \label{exe: symbolic matrix refined one-shot 4,1,2,M}
Consider the setting of Example \ref{exe: refined one-shot 4,1,2,M}, where $N=4$ and $r=2$. For $M = 2$ we get the symbolic matrix 
\begin{align} 
 S_2 = \begin{pmatrix}
1 & 1 & 2 & 2
\end{pmatrix}
.\end{align}

Following the procedure in Definition \ref{def: symbolic matrix} the symbolic matric $S_3$ is given by
\[ S_3 = \left( \begin{array}{cccc}
1 & 1 & 2 & 2 \\
1 & 2 & 2 & 1 \\
\multicolumn{4}{c}{A} 
\end{array} \right) .\]

The set $B = [2,S_2] = \{ (1,3), (1,4) \}$. Setting $b_1 = (1,3)$ and $b_2 = (1,4)$, then, $b_1 \prec b_2$. 

Since $\tau_1 (b_1) = (2,2)$ and $\tau_1 (b_2) = (2,3)$ it follows that $B_1 = \{(1,3), (2,2) \} $ and $B_2 = \{ (1,4) , (2,3) \}$.

Since $\pi (B_1) = \{ 3, 4 \}$ and $\pi (B_2) = \{2, 3 \}$, it follows that
\[ A = \left( \begin{array}{cccc}
3 &  &  & 3 \\
3 & 3 &  &  
\end{array} \right) \]
and therefore,
\begin{align} 
 S_3 = \begin{pmatrix}
1 & 1 & 2 & 2 \\
1 & 2 & 2 & 1 \\
3 &  &  & 3 \\
3 & 3 &  &  
\end{pmatrix}
.\end{align}

We will now determine $S_4$ using $S_3$. We will re-use the same variables as before so that they are in accordance with Definition \ref{def: symbolic matrix}. We have that

\[ S_3 = \left( \begin{array}{cccc}
1 & 1 & 2 & 2 \\
1 & 2 & 2 & 1 \\
3 &  &  & 3 \\
3 & 3 &  &  \\
1 & 2 & 2 & 1 \\
2 & 2 & 1 & 1 \\
& & 3 & 3 \\
3 & & & 3 \\
\multicolumn{4}{c}{A} 
\end{array} \right) .\]

The set $B = [3,S_3] = \{b_1 , \ldots, b_4 \}$ where $b_1 = (3,1)$, $b_2 = (3,4)$, $b_3 = (4,1)$, and $b_4 = (4,2)$.

The sets $B_1 = \{ (3,1) , (7,4) \}$, $B_2 = \{ (3,4) , (7,3) \}$, $B_3 = \{ (4,1) , (8,4) \}$, and $B_4 = \{ (4,2) , (8,1) \}$.

The sets $\pi(B_1) = \{ 1,4 \}$, $\pi(B_2) = \{ 3,4 \}$, $\pi(B_3) = \{ 1,4 \}$, and $\pi(B_4) = \{ 1,2 \}$, so that,
\[ A = \left( \begin{array}{cccc}
 & 4 & 4 &  \\
4 & 4 &  &  \\
  & 4 & 4 & \\
& & 4 & 4 
\end{array} \right) \]
and therefore,
\begin{align} 
 S_4 = \begin{pmatrix}
1 & 1 & 2 & 2 \\
1 & 2 & 2 & 1 \\
3 &  &  & 3 \\
3 & 3 &  &  \\
1 & 2 & 2 & 1 \\
2 & 2 & 1 & 1 \\
& & 3 & 3 \\
3 & & & 3 \\
 & 4 & 4 &  \\
4 & 4 &  &  \\
  & 4 & 4 & \\
& & 4 & 4 
\end{pmatrix}
.\end{align}

\end{example}

\subsection{Interpreting the Symbolic Matrix.} 

In this section we explain how to transform a symbolic matrix into a PIR scheme. 

\begin{definition}
Let $S_M$ be a symbolic matrix. The query structure of the symbolic matrix is obtained as follows. Each entry $k$ in the symbolic matrix $S_M$ represents $\binom{M}{k}$ $k$-queries, one for every combination of $k$ messages.
\end{definition}

\begin{example} \label{exe: symbolic to query}
The query structure of the symbolic matrix, $S_2$, in Example \ref{exe: symbolic matrix refined one-shot 4,1,2,M} is given in the following table.
\begin{table}[H]
\centering
\begin{tabular}{c c c c}
 Server $1$ &  Server $2$ & Server $3$ & Server $4$ \\
\toprule
$\bm{a}_1$ & $\bm{a}_2$ & $\bm{a}_3 + \bm{b}_3$ &  $\bm{a}_4 + \bm{b}_4$   \\
$\bm{b}_1$ & $\bm{b}_2$ & &
\end{tabular}
\caption{Query structure for $S_2$ in Example \ref{exe: symbolic matrix refined one-shot 4,1,2,M}.}
\end{table}

The query structure of the symbolic matrix, $S_3$, in Example \ref{exe: symbolic matrix refined one-shot 4,1,2,M} is given in the following table.

\begin{table}[H]
\centering
\begin{tabular}{c c c c}
 Server $1$ &  Server $2$ & Server $3$ & Server $4$ \\
\toprule
$\bm{a}_1$ & $\bm{a}_2$ & $\bm{a}_3 + \bm{b}_3$ &  $\bm{a}_5 + \bm{b}_5$   \\
$\bm{b}_1$ & $\bm{b}_2$ & $\bm{a}_4 + \bm{c}_3$ & $\bm{a}_6 + \bm{c}_5$ \\
$\bm{c}_1$ & $\bm{c}_2$ & $\bm{b}_4 + \bm{c}_4$ & $\bm{b}_6 + \bm{c}_6$ \\
\hdashline[2pt/1pt]
$\bm{a}_8$ & $\bm{a}_9 + \bm{b}_9$ & $\bm{a}_{11} + \bm{b}_{11}$ &  $\bm{a}_7$   \\
$\bm{b}_8$ & $\bm{a}_{10} + \bm{c}_9$ & $\bm{a}_{12} + \bm{c}_{11}$ &  $\bm{b}_7$ \\
$\bm{c}_8$ & $\bm{b}_{10} + \bm{c}_{10}$ & $\bm{b}_{12} + \bm{c}_{12}$ &  $\bm{c}_7$ \\
\hdashline[2pt/1pt]
$\bm{a}_{13} + \bm{b}_{13} + \bm{c}_{13}$ & & &  $\bm{a}_{14} + \bm{b}_{14} + \bm{c}_{14}$   \\
$\bm{a}_{15} + \bm{b}_{15} + \bm{c}_{15}$ & $\bm{a}_{16} + \bm{b}_{16} + \bm{c}_{16}$ & &  

\end{tabular}
\caption{Query structure for $S_3$ in Example \ref{exe: symbolic matrix refined one-shot 4,1,2,M}.}
\label{tab: query structure s3}
\end{table}

\end{example}

We need now to specify the distribution of the queries. We do this recursively on the number of messages.

For $M=2$ messages we choose the distribution of the queries as in Lemma \ref{lem:ref}. 

For $M=3$, the query structure for the symbolic matrix $S_2$ is as follows.

\begin{table}[H]
\centering
\begin{tabular}{c c c c c c}
Server $1$ & $\cdots$ & Server $r$ & Server $r+1$ & $\cdots$ & Server $N$ \\
\toprule
 $\bm{a}_1$ & $\cdots$ & $\bm{a}_r$ & $\bm{a}_{r+1}+\bm{b}_{r+1} $ & $\cdots$ & $\bm{a}_{2N -1}+\bm{b}_{2N -1}$  \\
 $\bm{b}_1$ & $\cdots$ & $\bm{b}_r$ & $\bm{a}_{r+2}+\bm{c}_{r+1} $ & $\cdots$ & $\bm{a}_{2N}+\bm{c}_{2N -1}$ \\
 $\bm{c}_1$ & $\cdots$ & $\bm{c}_r$ & $\bm{b}_{r+2}+\bm{c}_{r+2} $ & $\cdots$ & $\bm{b}_{2N}+\bm{c}_{2N}$
\end{tabular}
\caption{Query structure for $S_2$ when $M=3$.}
\label{tab:partial lift}
\end{table}

We call the queries of the form $\bm{b}_* + \bm{c}_*$ leftovers. If we ignore the leftovers we can choose the $\bm{b}$'s with the same distribution as in the case of $M=2$, and choose the $\bm{c}$'s with the same distribution as the $\bm{b}$'s. In this way, all the $\bm{a}$'s in Table \ref{tab:partial lift} can be decoded via the decoding equations of the one-shot scheme.

We can do the same for each matrix $S_2, \ldots, \sigma^{r-1}(S_2)$. Thus, in the symbolic matrix $S_3$ each entry of value $2$ has a corresponding leftover query. Using the notation of Definition \ref{def: symbolic matrix} there is one leftover query per entry in the set $B$.

Each entry of value $3$ in the matrix $A$ corresponds to a query of the form $\bm{a}_*+\bm{b}_*+\bm{c}_*$. Consider the leftover queries corresponding to the entries in the set $B_i$. If we choose the leftover queries corresponding to $B_i$ together with the terms $\bm{b}_*+\bm{c}_*$ corresponding to $\{ (i,j) : A_{ij} = 3 \}$ with the same distribution as the $b$'s in Lemma \ref{lem:ref}, we can decode each one of the $a$'s in the $i$-th row of $A$.

The same procedure is carried out for $M=4,5,\ldots$

\begin{example}
The queries in Table \ref{tab: query structure s3} satisfy the following properties:

\begin{itemize}
\item The $\bm{a}_i \in \mathbb{V}_1$, $\bm{b}_i \in \mathbb{V}_2$, and $\bm{c}_i \in \mathbb{V}_3$.
\item The vectors $\bm{a}_1, \ldots, \bm{a}_{16} \in \mathbb{V}_1$ are chosen uniformly at random and such that $\langle \bm{D} , \bm{a}_1 \rangle, \ldots, \langle \bm{D} , \bm{a}_{16} \rangle$ are linearly independent, which in this case means that $\bm{a}_1, \ldots, \bm{a}_{16}$ are linearly independent.

\item The vectors $\bm{b}_1,\bm{b}_2,\bm{b}_3,\bm{b}_5 \in \mathbb{V}_2$, $\bm{b}_8,\bm{b}_9,\bm{b}_{11},\bm{b}_7 \in \mathbb{V}_2$, $\bm{c}_1,\bm{c}_2,\bm{c}_3,\bm{c}_5 \in \mathbb{V}_3$, and $\bm{c}_8,\bm{c}_9,\bm{c}_{11},\bm{c}_7 \in \mathbb{V}_3$ are chosen with the same distribution as $\bm{b}_1 , \ldots, \bm{b}_4$ in Example \ref{exe: refined one-shot 4,1,2,M}.

\item The vectors $\bm{b}_{13} + \bm{c}_{13}, \bm{b}_{10} + \bm{c}_{10},\bm{b}_{4} + \bm{c}_{4},\bm{b}_{14} + \bm{c}_{14} \in \mathbb{V}_2 + \mathbb{V}_3$ and $\bm{b}_{15} + \bm{c}_{15}, \bm{b}_{16} + \bm{c}_{16},\bm{b}_{12} + \bm{c}_{12},\bm{b}_{6} + \bm{c}_{6} \in \mathbb{V}_2 + \mathbb{V}_3$ are chosen with the same distribution as $\bm{b}_1 , \ldots, \bm{b}_4$ in Example \ref{exe: refined one-shot 4,1,2,M}.
\end{itemize}
\end{example}

To find the rate of the lifted scheme we need to count the number of entries in the symbolic matrix of a specific value.

\begin{proposition} \label{pro:count}
Let $S_M$ be the symbolic matrix of a one-shot scheme with co-dimension $r$. Then,
\[ \#(k,S_M) =  \left( \frac{N-r}{r} \right)^{k-1} r^{M-1} \hspace{10pt} 1 \leq k \leq M \]
\end{proposition}
\begin{proof}
It follows from the lifting operation that
\begin{align*}
\#(k,S_M) &= \frac{N-r}{r} \#(k-1,S_M) \\
&= \left( \frac{N-r}{r} \right)^{k-1} \#(1,S_M) \\
&= \left( \frac{N-r}{r} \right)^{k-1} r^{M-1}.
\end{align*}
\end{proof}

\begin{restate}[The Lifting Theorem] \label{teo:lift}
Let $\mathcal{Q}$ be a one-shot scheme of co-dimension $r$. Then, refining and lifting  $\mathcal{Q}$ gives an $(N,K,T,M)$-PIR scheme $\mathcal{Q}'$ with rate
\[ R_{\mathcal{Q}'} = \frac{(N-r)N^{M-1}}{N^M - r^M} = \frac{N-r}{N \left( 1 - \left(\frac{r}{N}\right)^M\right)} .\]
\end{restate}

\begin{proof}
Given the one shot-scheme $\mathcal{Q}$, we apply the refinement lemma to obtain a scheme with symbolic matrix $S_2$ as in Equation \ref{eq:S_2}. The scheme $\mathcal{Q}'$ is defined as the one with symbolic matrix \mbox{$S_M = \lift^{M-2} (S_2)$} as in \eqref{eq:symbolicmatrix}.\footnote{The power in the expression $\lift^{M-2} (S_2)$ denotes functional composition.} Privacy and correctness of the scheme follow directly from the privacy and correctness of the one-shot scheme.

Next, we calculate the rate $R_{\mathcal{Q}'} = \frac{L}{|{\mathcal{Q}'}^1|}$.
Each entry $k$ of $S_M$ corresponds to $\binom{M}{k}$ $k$-queries, one for each combination of $k$ messages. Thus, using Proposition \ref{pro:count},
\begin{align*}
|{\mathcal{Q}'}^1| &= \sum_{k=1}^M \#(k,S_M) \binom{M}{k} \\
&= r^{M-1} \sum_{k=1}^M \left( \frac{N-r}{r} \right)^{k-1} \binom{M}{k} \\
&= \frac{N^M - r^M}{N-r}
\end{align*}

To find $L$ we need to count the queries which query $\bm{W}_1$. The number of $k$-queries which query $\bm{W}_1$ is $\binom{M-1}{k-1}$. Thus,
\begin{align*}
L &= \sum_{k=1}^M \#(k,S_M) \binom{M-1}{k-1} \\
&= r^{M-1} \sum_{k=1}^M \left( \frac{N-r}{r} \right)^{k-1} \binom{M-1}{k-1} \\
&= N^{M-1}
\end{align*}

Therefore, $R_{\mathcal{Q}'} = \frac{(N-r)N^{M-1}}{N^M - r^M}$.

\end{proof}

\section{Refining and Lifting Known Schemes} \label{sec:applications}

In this section, we refine and lift known one-shot schemes from the literature.

We first refine and lift the scheme described in \mbox{Theorem $3$} of \cite{taje18}. In our notation, this scheme is a one-shot scheme with co-dimension $r=\frac{NK-N+T}{K}$. 

\begin{theorem} \label{teo:salim}
Refining and lifting the scheme presented in \mbox{Theorem $3$} of \cite{taje18} gives an $(N,K,T,M)$-PIR scheme $\mathcal{Q}$ with
\begin{align} \label{eq:salim}
R_{\mathcal{Q}} = \frac{(N+T).(NK)^{M-1}}{(NK)^M - (NK-N+T)^M} .
\end{align} 
\end{theorem}

Next, we refine and lift the scheme in \cite{freij2016private}. In our notation, this scheme has co-dimension $r=K+T-1$. Thus, we obtain the first PIR scheme to achieve the rate conjectured to be optimal\footnote{The optimality of the rate in Equation \ref{eq:hollanti} was disproven in  \cite{jafar2018conjecture} for some parameters. For the remaining range of parameters, the PIR schemes obtained here in Theorem~\ref{teo:holl}, through refining and lifting, achieve the best rates known so far in the literature.} in \cite{freij2016private} for MDS coded data with collusions. 

\begin{restate}
Refining and lifting the scheme presented in \cite{freij2016private} gives an $(N,K,T,M)$-PIR scheme, $\mathcal{Q}$, with
\begin{align} \label{eq:hollanti}
R_{\mathcal{Q}} = \frac{(N-K-T+1)N^{M-1}}{N^M - (K+T-1)^M} = \frac{1 - \frac{K+T-1}{N}}{1 - \left( \frac{K+T-1}{N} \right)^M} .
\end{align}
\end{restate}

The rate of the scheme in Theorem \ref{teo:salim}, Equation \ref{eq:salim}, is upper bounded by the rate of the scheme in Theorem \ref{teo:holl}, Equation \ref{eq:hollanti}, with equality when either $K=1$ or $N=K+T$.

Now, we consider the case of replicated data ($K=1$) on $N$ servers with at most $T$ collusions. A $T$-threshold linear secret sharing scheme \cite{CDN15} can be transformed into the following one-shot PIR scheme.

\begin{table}[H]
\centering
\begin{tabular}{c  c  c  c c c}
Server $1$ & \ldots & Server $T$ & Server $T+1$ & \ldots & Server $N$ \\
\toprule
$\bm{q}_1$ &  & $\bm{q_T}$ & $\bm{q_{T+1}}+\bm{a}_1 $ & & $\bm{q_N + a_{N-T}}$ 
\end{tabular}
\caption{The $T$-threshold linear secret sharing scheme.}
\end{table}

\begin{restate}
Refining and lifting a $T$-threshold linear secret sharing scheme gives an $(N,1,T,M)$-PIR scheme $\mathcal{Q}$ with capacity-achieving rate 
\[ R_{\mathcal{Q}} = \frac{(N-T)N^{M-1}}{N^M - T^M} .\]
\end{restate}

This scheme has the same capacity achieving rate as the scheme presented in \cite{sun2016capacity} but with less queries.

\ifCLASSOPTIONcaptionsoff
  \newpage
\fi



%

%

\end{document}